\documentclass[11pt,a4paper]{article}
\usepackage[english]{babel}
\usepackage{amsmath,amssymb,amsthm,enumerate,bbm}
\usepackage[margin=0.9in]{geometry}

\catcode`\<=\active \def<{
\fontencoding{T1}\selectfont\symbol{60}\fontencoding{\encodingdefault}}
\catcode`\>=\active \def>{
\fontencoding{T1}\selectfont\symbol{62}\fontencoding{\encodingdefault}}
\catcode`\|=\active \def|{
\fontencoding{T1}\selectfont\symbol{124}\fontencoding{\encodingdefault}}
\newcommand{\assign}{:=}
\newcommand{\nobracket}{}
\newcommand{\nocomma}{}
\newcommand{\noplus}{}
\newcommand{\tmem}[1]{{\em #1\/}}
\newcommand{\tmop}[1]{\ensuremath{\operatorname{#1}}}
\newcommand{\tmstrong}[1]{\textbf{#1}}
\newcommand{\tmtextit}[1]{{\itshape{#1}}}
\newenvironment{enumerateroman}{\begin{enumerate}[i.] }{\end{enumerate}}
\newtheorem{definition}{Definition}
\newtheorem{question}{Question}
\newtheorem{remark}{Remark}
\newtheorem{example}{Example}
\newtheorem{proposition}{Proposition}
\DeclareFontFamily{U}{mathx}{\hyphenchar\font45}
\DeclareFontShape{U}{mathx}{m}{n}{
      <5> <6> <7> <8> <9> <10>
      <10.95> <12> <14.4> <17.28> <20.74> <24.88>
      mathx10
      }{}
\DeclareSymbolFont{mathx}{U}{mathx}{m}{n}
\DeclareFontSubstitution{U}{mathx}{m}{n}
\DeclareMathAccent{\widecheck}{0}{mathx}{"71}
\DeclareMathAccent{\wideparen}{0}{mathx}{"75}

\begin{document}

\title{Some Properties of a Class of QM Phase Space Measures}

\author{Kurt Pagani}

\maketitle

\begin{abstract}
  We have a look at the probability measures induced by Schr\"odinger wave functions on phase space. 
	
\end{abstract}

\section{Introduction}

Let us consider the quantum mechanical energy of a particle in a potential
$V:$
\begin{equation}
  \mathcal{E} ( \phi ) = \frac{\hbar^{2}}{2m} \int_{\mathbbm{R}^{n}} | \nabla
  \phi ( x ) |^{2}  d x  \noplus + \int_{\mathbbm{R}^{n}} V ( x )   | \phi ( x
  ) |^{2}  d x \label{eq1}
\end{equation}
whereby $\phi \in L^{2} ( \mathbbm{R}^{n} ) \cap \{ \| \phi \| =1 \}$. Usually
$n=3$, but we leave it unspecified. Using the Fourier transformation on $L^{2}$
in the form
\[ \hat{f} ( k ) = ( 2 \pi )^{- \frac{n}{2}} \int_{\mathbbm{R}^{n}} f ( x ) 
   e^{-i  \langle k,x \rangle}  d x \]
we may rewrite $\left( \ref{eq1} \right)$ to
\begin{equation}
  \mathcal{E} ( \phi ) = \frac{\hbar^{2}}{2m} \int_{\mathbbm{R}^{n}} | k |^{2}
  | \hat{\phi} ( k ) |^{2}  d k  \noplus + \int_{\mathbbm{R}^{n}} V ( x )  
  | \phi ( x ) |^{2}  d x, \label{eq2}
\end{equation}
where
\[ \| \phi \| \assign \left( \int_{\mathbbm{R}^{n}} | \phi ( x ) |^{2}  d x
   \right)^{1/2} = \| \hat{\phi} \| =1. \]
Multiplying the first integral in $\left( \ref{eq2} \right)$ by $\| \phi
\|^{2}$ and the second by $\| \hat{\phi} \|^{2}$ we get (assuming Fubini's
theorem is applicable):
\begin{equation}
  \mathcal{E} ( \phi ) = \int_{\mathbbm{R}^{n} \times \mathbbm{R}^{n}} \left(
  \frac{\hbar^{2} | k |^{2}}{2m} +V ( x ) \right) | \phi ( x ) |^{2}   |
  \hat{\phi} ( k ) |^{2}  d x d k \label{eq3}
\end{equation}
which has the general form
\begin{equation}
  \mathcal{E} ( \phi ) = \int_{\Gamma_{n}} \mathcal{H} ( x, \hbar  k )   |
  \phi ( x ) |^{2}   | \hat{\phi} ( k ) |^{2}  d x d k = \int_{\Gamma_{n}}
  \mathcal{H} ( x, \hbar  k ) d \mu_{\phi} ( x,k ) .  \label{eq4}
\end{equation}
when introducing the classical Hamilton function $\mathcal{H} ( x,p ) =
\frac{p^{2}}{2m} +V ( x )$ and denoting the $2n-$dimensional phase space by
$\Gamma_{n} =\mathbbm{R}^{n} \times \mathbbm{R}^{n}$. The {\tmem{probability}}
measure
\[ d \mu_{\phi}  =  | \phi ( x ) |^{2}   | \hat{\phi} ( k ) |^{2}  d x d k,
\]
whose interpretation seems rather obvious, is absolutely continuous with
respect to the canonical phase space measure (Lebesgue) $d \Gamma =d x  \wedge
d k$ and has an integrable density (Radon-Nikodym)
\[ \frac{d \mu_{\phi}}{d  \Gamma} = | \phi ( x ) |^{2}   | \hat{\phi} ( k )
   |^{2}   \in L^{1} ( \Gamma_{n} ) \nocomma , \; \forall \phi \in L^{2} (
   \mathbbm{R}^{n} ) \cap \{ \| \phi \| =1 \} . \]
Equation $\left( \ref{eq4} \right)$ has just the form of an ordinary
expectation value and nothing but the constant $\hbar$ reminds to quantum
mechanics. Therefore let us forget the meaning of $\hbar$ for the moment and
merely consider it as a positive constant linking the dimension of Fourier
space $( k )$ to momentum space $( p )$ by
\[ p= \hbar  k. \]
Then we may put the cart before the horse by asking about the possibilities of
a {\tmem{probability theory}} on $\Gamma$ based on a Hamilton function
$\mathcal{H}$ and classical mechanics. But before, some questions that spring
to mind:

\begin{question}
  What is the minimum or more generally, what is the Euler equation to $\left(
  \ref{eq4} \right)$ if $\mathcal{H}$ is of general form? \ 
\end{question}

It is obvious that the measure $d \mu_{\phi_{0}}$ corresponding to a ground
state $\phi_{0}$ should be strongly supported where $\mathcal{H}$ is small
(e.g. on $\{ \mathcal{H}<0 \}$ for bound states), therefore

\begin{question}
  What can be said about the densities $\rho_{\phi} ( x,k ) \assign | \phi ( x
  ) |^{2}   | \hat{\phi} ( k ) |^{2}$ locally or - in case $\rho_{\phi}$ is
  continuous - pointwise?
\end{question}

And finally

\begin{question}
  Is there any explanation of the mysterious thumb rule $N_{\mathcal{H}} (
  \Lambda ) \sim | \mathcal{H}< \Lambda |  $for the distribution of states on
  grounds of $\left( \ref{eq4} \right) ?$
\end{question}

The last question is also connected with the quantity $\sum_{j} \int_{\Gamma}
\mathcal{H} d \mu_{\phi_{j}}$, where $\phi_{j}$ are eigenstates of the
corresponding ``Hamiltonian''. There are many more open questions of course,
like the relation to Wigner-Weyl, Moyal and other representations which will
not be touched here. Note that we use the notion of phase space - slightly
careless - for $( x,k ) - \tmop{space}$ as well as for $( q,p ) -
\tmop{space}$ and that we omitted the time dependency of the states $\phi$
which can be introduced if needed when replacing $\phi ( x )$ by $\phi ( x,t
)$ and $\Gamma_{n}$ by $\Gamma_{n} \times \mathbbm{R}$. \ \

\section{Possibilities of Phase Space Probabilities}

Given a classical Hamilton function $\mathcal{H} ( q,p,t )$, continuous on
$\mathbbm{R}^{2n+1}$ (for simplicity), and a probability measure $\mu$ on
$\Gamma_{n} =\mathbbm{R}^{n} \times \mathbbm{R}^{n}$ we set
\[ \mathfrak{E} ( \mu ,t ) = \int_{\Gamma_{n}} \mathcal{H} ( q,p,t )  d \mu (
   q,p ) \assign \mu ( \mathcal{H} ) , \hspace{1.2em} \mu ( 1 ) =1. \]
The interpretation of $\mu$ is clear: $\mu ( \chi_{A} ) =$probability to find
the ``particle'' in states $( q' ,p' )$ lying in the set $A \subset
\Gamma_{n}$, where the time $t $ is kept fix. Admitting general Radon measures
means that for instance
\[   \delta ( q-q_{0} ) \otimes \delta ( p-p_{0} ) \]
is allowed, so that
\[ \mathfrak{E} ( \mu ,t ) =\mathcal{H} ( q_{0} ,p_{0} ,t ) \]
and minimizing $\mathfrak{E} ( \mu ,t )$ results in finding the minimum of
$\mathcal{H}$. This is indeed not very interesting, therefore we should
restrict the admissible measures, for example:
\[ d \mu ( q,p ) = \rho ( q,p )  d q d p \]
where $\rho \in L^{1} ( \Gamma_{n} )$, $\rho \geqslant 0$ and
\[ \int_{\Gamma_{n}} \rho ( q,p )  d q d p=1. \]
\begin{proposition}
  Suppose $\mathcal{H} ( q,p )$ may be written as $T ( p ) +V ( q )$, then for
  any density $\rho$ satisfying the conditions above exist functions $\psi ,
  \phi \in L^{2} ( \mathbbm{R}^{n} )$ such that $\| \psi \| = \| \varphi \|
  =1$ and
  \[ \mathfrak{E} ( \mu ,t ) = \int_{\Gamma_{n}} \mathcal{H} ( q,p,t )   \rho
     ( q,p )  d q d p= \int_{\mathbbm{R}^{n}} T ( p,t ) | \varphi ( p ) |^{2} 
     d p+ \int_{\mathbbm{R}^{n}} V ( q,t )   | \psi ( q ) |^{2}  d q. \]
  Moreover, it holds:
  \[ \mathfrak{E} ( \mu ,t ) = \int_{\Gamma_{n}} \mathcal{H} ( q,p,t ) | \psi
     ( q ) |^{2}   | \varphi ( p ) |^{2}  d q d p. \]
  
\end{proposition}

\begin{proof}
  Set $f ( q ) = \int_{\mathbbm{R}^{n}} \rho ( q,p )  d p$ and $g ( p ) =
  \int_{\mathbbm{R}^{n}} \rho ( q,p ) d q$, then by Fubini $f,g \in L^{1} (
  \mathbbm{R}^{n} )$ and
  \[ \int_{\mathbbm{R}^{n}} f ( q )  d q= \int_{\mathbbm{R}^{n}} g ( p )  d
     p=1 \]
  and
  \[ \mathfrak{E} ( \mu ,t ) = \int_{\Gamma_{n}} \mathcal{H} ( q,p,t )   \rho
     ( q,p )  d q d p= \int_{\mathbbm{R}^{n}} T ( p,t ) g ( p )  d p+
     \int_{\mathbbm{R}^{n}} V ( q,t )  f ( q )  d q. \]
  Since $f,g$ are non-negative there are measurable functions $\psi , \varphi
  \in L^{2} ( \mathbbm{R}^{n} ,\mathbbm{C} )$ with the stated properties.
\end{proof}

\begin{remark}
  Is there a physically justifiable reason requiring that $\psi$ and $\varphi$ are
  connected via Fourier transform? If so, then the stationary Schr{\"o}dinger like
  equation would result as the Euler equation of the functional
  $\mathfrak{E} ( \mu ,t )$ when dealing as outlined in the introduction. 
	We will of course avoid any kind of speculation, however, one can try to extract
	as many properties out of this fact and compare it to other possible relations.
\end{remark}

We recall some properties of the Fourier transformation:
\[ A \in \tmop{SL} ( n ) \Rightarrow \widehat{R_{A} f}  = R_{A^{-T}} \hat{f}  
\]
\[ \epsilon >0  \Rightarrow   \widehat{S_{\varepsilon} f} = \varepsilon^{-n} 
   S_{1/ \varepsilon} \hat{f} \]
\[ a \in \mathbbm{R}^{n}   \Rightarrow   \widehat{T_{a} f} =e^{ i \langle a,
   \cdot \rangle} \hat{f} \]
where the operators $R_{A} ,S_{\varepsilon}$ and $T_{a}$ are rotation by $A$,
dilation by $\varepsilon$ and translation by the vector $a$. If $A \in O ( n
)$ then $A^{-T} := ( A^{-1} )^{T} =A$. In fact, it was shown by Hertle
{\cite{AH1982}} that if a continuous operator $F:\mathcal{D} ( \mathbbm{R}^{n}
) \rightarrow \mathcal{D}' ( \mathbbm{R}^{n} )$ satisfies the three relations
above for any $A \in O ( n )$, $\varepsilon >0$ and $a \in \mathbbm{R}^{n}$,
then it is a constant multiple of the Fourier transform. For $n=1$ it was
shown by Cooper {\cite{JC1970}} that any linear operator on $L^{2} (
\mathbbm{R} )$ which intertwines translations and modulations must be a
Fourier transform. Some newer results {\cite{AV2008}} characterize FT even
without the assumption of linearity.

\subsection{Properties of the measures $d \mu_{\phi}$}

Let $\mathcal{M}^{\flat} ( \Gamma_{n} )$ denote the set of all Radon
measures of the form
\[ \mu [ \phi ] ( f ) = \int_{\mathbbm{R}^{n} \times \mathbbm{R}^{n}} f ( x,k
   )   | \phi ( x ) |^{2}   | \hat{\phi} ( k ) |^{2}  d x d k  \]
where $\phi \in L^{2} ( \mathbbm{R}^{n} ) \cap \{ \| \phi \| =1 \}$, $\forall
f \in C^{0} ( \Gamma_{n} ,\mathbbm{R} )$, $\Gamma_{n} =\mathbbm{R}^{n} \times
\mathbbm{R}^{n} .$ All these measures are positive and bounded because
\[ \mu [ \phi ] ( f ) \geqslant 0  \tmop{if}  f \geqslant 0  \tmop{and}   \mu
   [ \phi ] ( 1 ) =1. \]
Therefore every bounded measurable function $f: \Gamma_{n} \rightarrow
\mathbbm{R} \cup \{ \pm \infty \}$ is integrable, in particular, every bounded
semi-continuous $f$ is integrable. Instead of using measures one can also
think of $\mu [ \phi ]$ as
\[ \langle \phi \otimes \hat{\phi} | f | \phi \otimes \hat{\phi} \rangle \]
with $\langle \phi \otimes \hat{\phi} | | \phi \otimes \hat{\phi} \rangle =
\langle \phi , \phi \rangle \langle \hat{\phi} , \hat{\phi} \rangle = \| \phi
\|^{2}   \| \hat{\phi} \|^{2} =1$ on $L^{2} ( \Gamma^{n} ) \simeq L^{2} (
\mathbbm{R}^{n} ) \otimes L^{2} ( \mathbbm{R}^{n} )$.

\begin{proposition}
  \label{p6}With the notation above, every $\mu_{\phi} \in \mathcal{M}^{\flat}
  ( \Gamma_{n} )$ has the following properties:
  \begin{enumerateroman}
    \item $\mu [ R_{A} \phi ] ( f ) = \frac{1}{| \det ( A ) |^{2}} \mu [ \phi
    ] \left( R_{A^{\mbox{-} 1} \otimes A^{T}}  f \right)  , \forall A \in
    \tmop{GL} ( n )$
    
    \item $\mu [ S_{\lambda} \phi ] ( f ) = \lambda^{\mbox{-} 2n} \mu [ \phi ]
    \left( S_{\frac{1}{\lambda} \otimes \lambda} f \right) , \; \forall
    \lambda >0$
    
    \item $\mu [ T_{a} \phi_{} ] ( f ) = \mu [ \phi ] \left( T_{\mbox{-} a
    \otimes 0} f_{} \right)$, $\forall a \in \mathbbm{R}^{n}$
  \end{enumerateroman}
\end{proposition}

\begin{proof}
  Let $A \in \tmop{GL} ( n )$, then
  \[ \mu [ R_{A} \phi ] ( f ) = \int_{\Gamma_{n}} f ( x,k ) | \phi ( A x )
     |^{2}   | \widehat{R_{A} \phi} ( k ) |^{2}  d x d k  \]
  \[ = | \det ( A ) |^{-2} \int_{\Gamma_{n}} f ( x,k ) | \phi ( A x ) |^{2}  
     | \hat{\phi} ( A^{-T}  k ) |^{2}  d x d k  \]
  \[ \int_{\Gamma_{n}} f \left( A^{\mbox{-} 1} \xi ,A^{T} \eta \right) | \phi ( \xi ) |^{2}
     | \hat{\phi} ( \eta ) |^{2}   \frac{d \xi d \eta}{| \det ( A ) |^{2}} =
     \frac{1}{| \det ( A ) |^{2}} \mu [ \phi ] \left( R_{A^{\mbox{-} 1}
     \otimes A^{T}}  f \right) \]
  where the coordinate transformation $\xi =A x, \eta =A^{\mbox{-} T}  k$ was
  used. Note that $A^{-T}   \tmop{means}   ( A^{-1} )^{T}  $. The proof of $(
  \tmop{ii} ) , ( \tmop{iii} )$ goes along the same lines.

\end{proof}

The special linear group $\tmop{SL}_{} ( n,\mathbbm{R} )$ induces a subgroup
of the symplectic group $\tmop{Sp} ( 2n,\mathbbm{R} )$ as follows:
\[ \mathcal{R}_{A} = \left(\begin{array}{cc}
     A^{\mbox{-} 1} & 0\\
     0 & A^{T}
   \end{array}\right) , \; \forall A \in \tmop{SL} ( n ) . \]
That $\mathcal{R}_{A} \in \tmop{Sp} ( 2n,\mathbbm{R} )$ follows from
\[ \mathcal{R}_{A}^{T}   \Omega  \mathcal{R}_{A} = \left(\begin{array}{cc}
     A^{\mbox{-} T} & 0\\
     0 & A^{}
   \end{array}\right) \left(\begin{array}{cc}
     0 & A^{T}\\
     \mbox{-} A^{\mbox{-} 1} & 0
   \end{array}\right) = \left(\begin{array}{cc}
     0 & I_{n}\\
     \mbox{-} I_{n} & 0
   \end{array}\right) =  \Omega \]
where $I_{n}$ is the unit matrix in $\mathbbm{R}^{n}$. The special case $A \in
\tmop{SO} ( n )$ gives $\mathcal{R}_{A} = \tmop{diag} [ A,A ]$. Moreover we
have
\[ \mathcal{R}_{A}  \mathcal{R}_{B}  = \mathcal{R}_{B A} \in \tmop{Sp} (
   2n,\mathbbm{R} ) \]
\begin{proposition}
  $\mathcal{M}^{\flat} ( \Gamma_{n} )$ is {\tmstrong{invariant}} under $\{
  \mathcal{R}_{A} :A \in \tmop{SL} ( n,\mathbbm{R} ) \} \vartriangleleft
  \tmop{Sp} ( 2n,\mathbbm{R} )$.
\end{proposition}

\begin{proof}
  This is an immediate consequence of Proposition (\ref{p6}$i$). Indeed,
  \[ \mu [ \phi ] \left( R_{A^{\mbox{-} 1} \otimes A^{T}}  f \right) = \mu [
     \phi ] ( \mathcal{R}_{A} f ) = \mu [ R_{A} \phi ] ( f ) \]
  since $\det ( A ) =1$ if $A \in \tmop{SL} ( n,\mathbbm{R} ) .$
\end{proof}

What about the full group? Let
\[ M = \left(\begin{array}{cc}
     A^{} & B\\
     C & D
   \end{array}\right) , \; A,B,C,D \in \mathfrak{M}_{n \times n} \]
then $M$ being symplectic is equivalent to the conditions:
\[ A^{T} D-C^{T} B=I_{n} \]
\[ A^{T} C=C^{T} A \]
\[ D^{T} B=B^{T} D \]
Now
\[ \mu [ \phi ] ( R_{M} f ) = \int_{\Gamma_{n}} f ( A x+B k,C x+D k )   | \phi ( x ) |^{2} 
   | \hat{\phi} ( k ) |^{2}  d x d k \]
and using
\[ M^{-1} = \Omega^{-1} M^{T}   \Omega = \left(\begin{array}{cc}
     D^{T} & -B^{T}\\
     -C^{T} & A^{T}
   \end{array}\right) \]
we obtain
\[ \mu [ \phi ] ( R_{M} f ) = \int_{\Gamma_{n}} f ( \xi , \eta )   \vert  \phi 
  \left( D^{T} \xi \mbox{-} B^{T} \eta \right) \vert ^{2}   
	 \vert  \hat{\phi} \left(\mbox{-} C^{T} \xi +A^{T} \eta \right) 
	  \vert ^{2}  d  \xi  d  \eta . \]
Suppose for the moment that $A=D=0$, then the first condition requires $C=
\mbox{-} B^{-T}$, thus the integral above reduces to
\[ \int_{\Gamma_{n}} f ( \xi , \eta )   | \phi ( C^{-1}   \eta ) |^{2}   \vert  \hat{\phi}
   \left( \mbox{-} C^{T} \xi \right) \vert ^{2}  d  \xi  d  \eta \]
which cannot be of the form $\mu [ \phi' ] ( f )$ unless $\vert  \hat{\phi}
\left( \mbox{-} k \right) \vert ^{} = | \hat{\phi} ( k ) |^{}$. The latter
holds if one restricts to real or purely imaginary functions $\phi \nocomma$.

It is easily seen that if
\[ \int_{\Gamma_{n}} \frac{d \mu_{\phi}}{( 1+ | x | + | k | )^{N}} < \infty \]
for some $N \in \mathbbm{Z}_{+}$, then the measure $\mu_{\phi}$ extends to a
tempered distribution. Most of the common {\tmem{uncertainty}}
{\tmem{relations}} (Heisenberg, Weil, Hardy ...) are based on the fact that a
function and its Fourier transform are in some sense antagonists. Therefore it
is expected that $\mu_{\phi}$ cannot be too localizing. Of course, if $\phi$
has compact support then $\hat{\phi}$ is a real analytic function of $k$,
hence cannot vanish identically on an open subset of $\mathbbm{R}^{n}$, and so
it is impossible that $\tmop{supp} \{ \mu_{\phi} \}$ is compact. In fact, a
nonzero $\phi$ cannot vanish outside a set of finite Lebesgue measure while
$\hat{\phi}$ does the same for (possibly) another set of finite measure, as
was proved by Benedicks {\cite{MB1985}}. Quantitative results to this fact
were obtained by Nazarov and Steiner {\cite{AS1974}} for example in the case
$n=1$ and generalized to $n \geqslant 1$ by Jaming {\cite{PJ2006}}. \

\begin{proposition}
  Suppose $N>2n$ and
  \[ \int_{\Gamma_{n}} \mathcal{H} ( x,k )  d \mu_{\phi}  <  \infty \]
  for some normalized $\phi \in L^{2} ( \mathbbm{R}^{n} )$, where the function
  $\mathcal{H}$ grows like
  \[ \mathcal{H} ( x,k ) \sim \frac{e^{2 | \langle x,k \rangle |}}{( 1+ | x |
     + | k | )^{N}} , \hspace{1.8em} | x |^{2} + | k |^{2}  > R \gg 1 \]
  then $\phi$ is of the form
  \[ \phi ( x ) =P ( x )   \exp ( - \langle x,A x \rangle ) , \hspace{1.2em} A
     \in O ( n )_{+} , \]
  where $P$ is a polynomial with $\deg ( P ) < \frac{( N-n )}{2}$.

\end{proposition}

This shows that the measures $\mu_{\phi}$ cannot handle the case where very
rapid decreasing in both variables $x,k$ would be necessary to provide
finiteness unless $\phi$ is a Gaussian function times a polynomial. Moreover
it demonstrates how subtle the balance is relative to such functions. The
proof is a simple corollary of a (recent) generalized version of the
Beurling-H{\"o}rmander principle {\cite{BDJ2001}}:

For $N \geqslant 0$ and $\varphi \in L^{2} ( \mathbbm{R}^{n} )$,
\[ \int_{\mathbbm{R}^{n} \times \mathbbm{R}^{n}} \frac{| \varphi ( x )
   \nobracket   | \hat{\varphi} ( k ) |}{( 1+ | x | + | k | )^{N}}  e^{|
   \langle x,k \rangle |}  d x d k <  \infty \]
if and only if
\[ \varphi ( x ) =P ( x )  e^{- \langle x,A x \rangle} \]
for a positive definite symmetric matrix $A$ and a polynomial $P$ of degree
smaller than $\frac{N-n}{2} .$ Moreover, if $N<n$, then $\varphi \equiv 0$.
\[  \]

\section{General Hamilton Functions}

Let $\mathcal{H} ( q,p,t ) \in C ( \Gamma_{n} )$ for any fixed $t$. The
functional
\[ \mathcal{E} ( \phi ,t ) = \int_{\Gamma_{n}} \mathcal{H} ( x, \hbar k,t )  d
   \mu_{\phi} ( x,k ) \]
is well defined, however, if $\mathcal{H}$ is not bounded it may assume values
in $\bar{\mathbbm{R}} =\mathbbm{R} \cup \{ \pm \infty \}$. Since time $t$
plays no role in the following, we will omit it.

\begin{proposition}
  Suppose $\phi_{0} \in L^{2} ( \mathbbm{R}^{n} ) \cap \{ \| \phi \| =1 \}$
  satisfies
  \[ \mathcal{E} ( \phi_{0} ) =E_{0} = \inf \left\{ \int_{\Gamma_{n}}
     \mathcal{H} ( x, \hbar k )  d \mu_{\phi} ( x,k )  : \phi \in L^{2} (
     \mathbbm{R}^{n} ) , \| \phi \| =1  \right\} \in \mathbbm{R} \]
  then it is a (distributional) solution to
  \[ F_{\phi} ( x )   \phi ( x ) + {\cal{F}}^{-1}({G_{\phi} ( k )  \hat{\phi} ( k )}
     ) ( x ) =2 E_{0} \, \phi ( x ) \]
  where \footnote{${\cal{F}}^{-1}$ = inverse FT (lacking a reasonable {\tt widecheck}  symbol).}
  \[ F_{\phi} ( x ) = \int_{\mathbbm{R}^{n}} \mathcal{H} ( \hbar k,x )   |
     \hat{\phi} ( k ) |^{2}  d k, \hspace{2.4em} G_{\phi} ( k ) \assign
     \int_{\mathbbm{R}^{n}} \mathcal{H} ( \hbar k,x ) | \phi ( x ) |^{2}  d x.
  \]
\end{proposition}

\begin{proof}
  The formal variation with a Lagrange multiplier $\lambda$ is straightforward:
  \[ \delta \mathcal{E}= \langle \delta \phi ,F ( x ) \phi \rangle + \langle F
     \phi , \delta \phi \rangle + \langle \delta \hat{\phi} ,G ( k )
     \hat{\phi} \rangle + \langle G ( k ) \hat{\phi} , \delta \hat{\phi}
     \rangle = \lambda ( \langle \delta \phi , \phi \rangle + \langle \delta
     \hat{\phi} , \hat{\phi} \rangle + \tmop{cc} . ) \]
  where $\delta \phi \in C_{0}^{\infty} ( \mathbbm{R}^{n} )$. Since the
  Fourier transform is a unitary isomorphism on $L^{2} ( \mathbbm{R}^{n} )$
  one may shift the perturbations to the left:
  \[ \langle \delta \phi ,F \phi + {\cal{F}}^{-1}{G \hat{\phi}} \rangle =2 \lambda
     \langle \delta \phi , \phi \rangle \]
  so that
  \[ F \phi + {\cal{F}}^{-1}{G \hat{\phi}} =2 \lambda \phi   \; \tmop{in} 
     \mathcal{D}' ( \mathbbm{R}^{n} ) . \]
  Since $L^{2} ( \mathbbm{R}^{n} ) \simeq \{ T \in \mathcal{D}' (
  \mathbbm{R}^{n} ) : | T ( \phi ) | \leqslant c_{T} \| \phi \| \}$ we get
  \[ \left\langle \phi_{0} ,F \phi_{0} + {\cal{F}}^{-1}{G \widehat{\phi_{0}}}
     \right\rangle =2 \lambda \| \phi_{0} \|^{2} =2 \lambda \]
  that is
  \[ \langle \phi_{0} ,F \phi_{0} \rangle + \left\langle \phi_{0} , {\cal{F}}^{-1}{G
     \widehat{\phi_{0}}} \right\rangle =\mathcal{E} ( \phi_{0} ) + \langle
     \widehat{\phi_{0}} ,G \widehat{\phi_{0}} \rangle =2\mathcal{E} ( \phi_{0}
     ) \Rightarrow   \lambda =\mathcal{E} ( \phi_{0} ) . \]
  
\end{proof}

Clearly, $\langle \phi ,F_{\phi}   \phi \rangle  + \langle \phi ,
{\cal{F}}^{-1}{G_{\phi}   \hat{\phi}} \rangle = \langle \phi ,F_{\phi}   \phi \rangle
+ \langle \hat{\phi} ,G_{\phi}   \hat{\phi} \rangle =2E_{0}$, and
\[ \langle \phi ,F_{\phi}   \phi \rangle = \langle \hat{\phi} ,G_{\phi}  
   \hat{\phi} \rangle =E_{0} \]
which explains the ``doubling'' of the energy.

\begin{example}
  $\mathcal{H} ( \hbar k,x ) = \frac{\hbar^{2}}{2m} k^{2} +V ( x ) \Rightarrow
  F_{\phi} ( x ) =V ( x ) +E_{\tmop{kin}}$,\;$G_{\phi} ( k ) =
  \frac{\hbar^{2}}{2m} k^{2} +E_{\tmop{pot}}$
\end{example}
\[ \Rightarrow  - \frac{\hbar^{2}}{2m} \Delta \phi_{0} +V \phi_{0} =E \phi_{0}
\]
As expected, the Schr{\"o}dinger equation results. What if $\mathcal{H}$ is
not additive separable?

\begin{example}
  $\mathcal{H} (   \hbar k,x ) = c \sqrt{m^{2} c^{2} + ( \hbar k-e A ( x )^{}
  )^{2}} +e \Phi$.
  \[ F_{\phi} ( x ) = \int_{\mathbbm{R}^{n}} \left( c \sqrt{m^{2} c^{2} + (
     \hbar k-e A ( x )^{} )^{2}} +e \Phi \right) | \hat{\phi} ( k ) |^{2}  d k
  \]
  \[ G_{\phi} ( k ) = \int_{\mathbbm{R}^{n}} \left( c \sqrt{m^{2} c^{2} + (
     \hbar k-e A ( x )^{} )^{2}} +e \Phi \right) | \phi ( x ) |^{2}  d x \]
  which leads to a non-linear equation $F_{\phi}   \phi \noplus +
  {{\check{G}_{\phi}}} \star \phi =2E \phi$, which looks more like to a Hartree-Fock - than to a ``Schr{\"o}dinger'' equation. Does the wavefunction
  concept make sense here at all?
\end{example}

\subsection{Phase space volumes}

Let $F:\mathbbm{R} \rightarrow \mathbbm{R}$ be a convex function, then by
Jensen's inequality:
\[ F \left( \int_{\Gamma_{n}} f d \mu_{\phi} \right) \leqslant
   \int_{\Gamma_{n}} ( F \circ f )  d \mu_{\phi} \]
and
\[ F \left( \frac{1}{N} \sum_{j=1}^{N} \int_{\Gamma_{n}} f d \mu_{\phi_{j}}
   \right) \leqslant \frac{1}{N} \sum_{j=1}^{N} F \left( \int_{\Gamma_{n}} f d
   \mu_{\phi_{j}} \right) \leqslant \frac{1}{N} \sum_{j=1}^{N}
   \int_{\Gamma_{n}} ( F \circ f )  d \mu_{\phi_{j}} . \]
If $F$ is concave, the inequality signs have to be reversed.

\begin{definition}
  Let $N \in \mathbbm{Z}_{+} \nocomma$, $f \in C ( \Gamma_{n} )$, then we set
  \[ \Sigma_{f} ( N ) =  \inf \left\{ \sum_{j=1}^{N} \int_{\Gamma_{n}} f ( x,k
     )  d \mu [ \phi_{j} ]  :  \langle \phi_{k} , \phi_{l} \rangle = \delta_{k
     l} \nocomma , \phi_{k} \in L^{2} ( \mathbbm{R}^{n} ) \right\} \]
\end{definition}

Let $\chi_{L} ( x )$ be the characteristic function of a cube with side length
$L$, then
\[ | \widehat{\chi_{}}_{L} ( k ) |^{2} = \left( \frac{2}{\pi} \right)^{n}
   \prod_{j=1}^{n} \frac{\sin^{2} \left( \frac{k_{j} L}{2} \right)}{k_{j}^{2}}
   \leqslant   \left( \frac{2}{n \pi} \sum_{j=1}^{n} \frac{1}{k_{j}^{2}}
   \right)^{n} . \]

\begin{proposition}
  Let $\varphi_{1} , \ldots , \varphi_{N}$ be an orthonormal system in $L^{2}
  ( \mathbbm{R}^{n} )$ such that $\tmop{supp} ( \varphi_{j} ) \subset \Omega$,
  and $\| \varphi_{j} \|_{\infty} \leqslant C$ for all $j=1, \ldots ,N$, where
  $\Omega$ is an open bounded subset of $\mathbbm{R}^{n}$. Then
  \[ \sum_{j=1}^{N} \int_{\Gamma_{n}} f ( x,k )  d \mu [ \varphi_{j} ]
     \leqslant C  \frac{| \Omega |}{( 2 \pi )^{n}} \int_{\Omega \times
     \mathbbm{R}^{n}} f ( x,k )  d x d k, \]
  for all $f \in C ( \Gamma_{n} ) ,f \geqslant 0$. The right integral above
  can be infinite, of course.

\end{proposition}

\begin{proof}
  With $d m_{n} ( x ) = ( 2 \pi )^{\mbox{-} n/2}  d x$, we get
  \[ \sum_{j=1}^{N} | \varphi_{j} ( x ) |^{2}   | \hat{\varphi}_{j} ( k )
     |^{2}  =  \sum_{j=1}^{N} | \varphi_{j} ( x ) |^{2}   \vert 
     \int_{\mathbbm{R}^{n}} \varphi_{j} ( y ) e^{- i  \langle k ,y \rangle}  d
     m ( y ) \vert ^{2} \leqslant C  \frac{  \chi_{\Omega} ( x )}{( 2 \pi
     )^{n}} \sum_{j=1}^{N} \vert  \left\langle \varphi_{j} ,e^{\mbox{-} i 
     \langle k, \cdot \rangle} \right\rangle \vert ^{2}   \]
  Bessel's inequality yields
  \[ C  \frac{  \chi_{\Omega} ( x )}{( 2 \pi )^{n}} \sum_{j=1}^{N} \vert 
     \left\langle \varphi_{j} ,e^{\mbox{-} i  \langle k, \cdot \rangle}
     \right\rangle \vert ^{2}   \leqslant C  \frac{  \chi_{\Omega} ( x )}{( 2
     \pi )^{n}} \left\| e^{\mbox{-} i  \langle k, \cdot \rangle}
     \right\|_{L^{2} ( \Omega )} = C  \chi_{\Omega} ( x )   \frac{| \Omega
     |}{( 2 \pi )^{n}} . \]
  
\end{proof}

For instance, any disjoint union of $N$ cubes $Q_{L} \subset \Omega$ gives:
\[ \sum_{j=1}^{N} \int_{\Gamma_{n}} f ( x,k )  d \mu [ \chi_{j} ] \leqslant  
   \frac{| \Omega |}{( 2 \pi L )^{n}} \int_{\cup Q_{L}^{} \times
   \mathbbm{R}^{n}} f ( x,k )  d x d k. \]
Thus for any $N \in \mathbbm{Z}_{+}$ we can set $\Omega = \Omega_{N} =
\cup_{j=1}^{N}  Q_{L}^{j}$, which yields
\[ \Sigma_{f} ( N ) \leqslant \sum_{j=1}^{N} \int_{\Gamma_{n}} f ( x,k )  d
   \mu [ \chi_{j} ] \leqslant   \frac{N}{( 2 \pi )^{n}} \int_{\Omega_{N}
   \times \mathbbm{R}^{n}} f ( x,k )  d x d k. \]
When we denote by $f_{1} ( x )$ the partial function
\[ \int_{\mathbbm{R}^{n}} f ( x,k )  d k, \]
assuming that it is finite, then one can increase $N$ and as a consequence
$\Omega_{N}$ as long as
\[ \frac{N}{( 2 \pi )^{n}} \int_{\Omega_{N}} f_{1} ( x )  d x  \]
stays finite in order to get a non trivial upper bound to $\Sigma_{f} ( N )$.
This leads to Berezin type inequalities {\cite{FAB1972}} for which there is a
wealth of literature (as there is for uncertainty principles). See e.g.
{\cite{KW2013}} and references therein. The usual procedure goes along the
lines of supposing that there are ``eigenvalues'' $\lambda_{j}$ and functions
$\phi_{j}$
\[ \lambda_{1} ( f ) \leqslant \ldots \leqslant \lambda_{N} ( f ) \]
such that
\[ \Sigma_{f} ( N ) = \sum_{j=1}^{N} \lambda_{j} ( f ) = \sum_{j=1}^{N}
   \int_{\Gamma_{n}} f d \mu_{\phi_{j}} , \]
then by Jensen's inequality
\[ N F  \left( \frac{1}{N} \Sigma_{f} ( N ) \right) = \sum_{j=1}^{N} F (
   \lambda_{j} ) = \sum_{j=1}^{N} F \left( \int_{\Gamma_{n}} f d
   \mu_{\phi_{j}} \right) \leqslant \sum_{j=1}^{N} \int_{\Gamma_{n}} F \circ f
   d \mu_{\phi_{j}} , \]
so that when $F$ is suitably chosen, one gets some bounds on the sum of
eigenvalues or the state density $N_{\lambda} ( f ) = \| \nu_{\lambda}
\|_{2}^{2} = \| \hat{\nu} \|_{2}^{2}$, where
\[ \nu_{\lambda} ( x ) = \sum_{\{ j: \lambda_{j} \leqslant \lambda \}}
   \phi_{j} ( x ) . \]
Another possibility is to use the fact that
\[ \widehat{P_{m} ( x )  e^{- \frac{1}{2} | x |^{2}}}   ( k )  =  ( -i )^{m} 
   P_{m} ( k )  e^{- \frac{1}{2} | k |^{2}} \]
for any homogeneous harmonic polynomial, so that

\[ d \mu [ \varphi ] =P_{m}^{2} ( x ) P^{2_{}}_{m} ( k )   e^{- | x^{2} | - |
   k |^{2}} \]
for such functions $\varphi$ if suitably normalized.

\subsection{Volume probability}

What can be said about
\[ \int_{\{ \mathcal{H} \leqslant \Lambda \}} d \mu [ \varphi ] \]
besides that its value is in $[ 0,1 ] ?$ Clearly, one has to suppose that $\{
( x,k ) \in \Gamma_{n} :\mathcal{H} ( x, \hbar k ) < \Lambda \}$ is measurable
which is usually the case. Assume that
\[ \{ \mathcal{H} \leqslant \Lambda \} \subset A \times B \subset
   \mathbbm{R}^{n} \times \mathbbm{R}^{n} \]
where $A,B$ are measurable sets, then
\[ \int_{\{ \mathcal{H} \leqslant \Lambda \}} d \mu [ \varphi ]   \leqslant
   \int_{\Gamma_{n}} \chi_{A} ( x ) \chi_{B} ( k )  d \mu [ \varphi ] =
   \int_{A} | \varphi ( x ) |^{2}  d x  \int_{B} | \hat{\varphi} ( k ) |^{2} 
   d k \]
and reverse if $A \times B \subset \{ \mathcal{H} \leqslant \Lambda \} .$
Therefore it is worthwhile to consider the functional
\[ J ( \varphi ,A,B ) = \int_{A} | \varphi ( x ) |^{2}  d x  \int_{B} |
   \hat{\varphi} ( k ) |^{2}  d k \]
for $\varphi \in L^{2} ( \mathbbm{R}^{n} ) \nocomma$, $\| \varphi \|_{2} =1.$
Using the properties in Proposition \ref{p6}, we get
\[ J ( S_{\lambda} \varphi ,A,B ) = \lambda^{-2n} J \left( \varphi , \lambda
   A, \frac{1}{\lambda} B \right) , \; J ( T_{a} \varphi ,A,B ) =J ( \varphi
   ,A+a,B ) , \]
so that any upper bound to $J$ is expected to depend on the product $| A |   |
B |$ only. In fact, using the two orthogonal projections
\[ ( P_{A} \varphi ) ( x ) =  \chi_{A} ( x ) \varphi ( x )   \tmop{and}   (
   \widehat{P_{B} \varphi} ) ( k ) = \chi_{B} ( k ) \hat{\varphi} ( k ) \]
it follows immediately by Plancherel that
\[ \| P_{A} P_{B} \|_{\tmop{HS}}^{2} = \int_{\Gamma_{n}} | \chi_{A}^{} ( x )
   |^{2}   | \widehat{\chi_{}}_{B} ( k-x ) |^{2} d x d k=
   \int_{\mathbbm{R}^{n}} | \chi_{A}^{} ( x ) |^{2}  d x  \int | \chi_{B}^{} (
   k ) |^{2}  d k= | A |   | B | , \]
for measurable sets $A,B$ with finite Lebesgue measure (note: HS means the
Hilbert-Schmitdt norm), so that when $\mathbbm{I}-P_{A} P_{B}$ is invertible
we get
\[ \| \varphi \|_{}^{2} \leqslant \| ( \mathbbm{I}-P_{A} P_{B} )^{-1} \|^{2} 
   \| ( P_{A}  P_{B'} +P_{A'}   ) \varphi \|^{2} , \; A' =\mathbbm{R}^{n}
   \left\backslash A, \, B' \right. =\mathbbm{R}^{n} \backslash B \nobracket ,
\]
which in turn gives (recall that 
$\vert\vert \cdot \vert\vert_{\tmop{HS}}  \leqslant \vert\vert
\cdot \vert\vert$)
\[ \| \varphi \|_{}^{2} \leqslant \frac{\| P_{A'} \varphi \|^{2} + \| P_{B'}
   \varphi \|^{2}}{\left( 1- \sqrt{| A |   | B |} \right)^{2}} \]
\[  \]
whenever $| A |   | B | <1.$ If we rewrite $J$ to
\[ \sqrt{J ( \varphi ,A,B )} = \sqrt{  ( 1- \| P_{A'} \varphi \|^{2} )   ( 1-
   \| P_{B'} \varphi \|^{2} )} \leqslant \frac{1}{2} ( 2- \| P_{A'} \varphi
   \|^{2} - \| P_{B'} \varphi \|^{2} ) \]
it follows
\[ \sqrt{J ( \varphi ,A,B )} \leqslant \frac{1}{2} ( 2- \| P_{A'} \varphi
   \|^{2} - \| P_{B'} \varphi \|^{2} ) \leqslant \frac{1}{2} \left( 2- \left(
   1- \sqrt{| A |   | B |} \right)^{2} \right) . \]
We will keep this in the following Proposition:

\begin{proposition}
  \label{p14}For measurable sets $A,B \subset \mathbbm{R}^{n}$ such that $| A
  |   | B |  <1$ it holds:
  \[ \sup \{ J ( \varphi ,A,B ) :  \varphi \in L^{2} ( \mathbbm{R}^{n} ) , \|
     \varphi \| =1 \} \leqslant \left( 1- \frac{1}{2} \left( 1- \sqrt{| A |  
     | B |} \right)^{2} \right)^{2} \]
\end{proposition}

The above argument was used by Amrein und Berthier in {\cite{AB1977}}, to show
that
\[ \dim ( P_{A}^{\perp} \cap P_{B}^{\perp} ) L^{2} ( \mathbbm{R}^{n} ) =
   \infty . \]
The restriction $| A |   | B | <1$, can be overcome when using a theorem of
Nazarov which was generalized by Jaming {\cite{PJ2006}} to $\mathbbm{R}^{n}
\nocomma$, $n \geqslant 1:$ There is a constant $C_{0}$ such that for all
$\varphi \in L^{2} ( \mathbbm{R}^{n} ) \cap \{ \| \phi \| =1 \}$
\begin{equation}
  \| P_{A'} \varphi \|^{2} + \| P_{B'} \varphi \|^{2} \geqslant  C_{0}   \exp
  \left( -C_{0}   \min \left( | A | | B | \nocomma ,w ( A ) | A
  |^{\frac{1}{n}} ,w ( B ) | B |^{\frac{1}{n}} \right) \right) \nocomma ,
  \label{Jaming}
\end{equation}
where $w ( A ) = \int_{\tmop{SO} ( n )} | \Pi_{m} ( A ) | d \nu_{n} ( m )
\nocomma , \Pi_{m} :A \rightarrow \tmop{span} ( m e_{1} )$, is an ``average''
width of $A.$ If $A$ is a ball then one obtains the diameter, for example.
Nazarov's original statment reads a bit simpler, however, we have to set
$n=1:$
\begin{equation}
  \int_{\mathbbm{R}^{}} | \varphi ( x ) |^{2}  d x  \leqslant   \alpha  e^{2
  \beta   | A |   | B |}   \left( \int_{\mathbbm{R} \backslash A \nobracket} |
  \varphi ( x ) |^{2}  d x + \int_{\mathbbm{R} \backslash B \nobracket} |
  \hat{\varphi} ( k ) |^{2}  d k \right) \label{Nazarov}
\end{equation}
for some constants $\alpha , \beta$ independent of $\varphi ,A,B.$ Needless to
say, both proofs are far from trivial. For our $J$ under consideration we
obtain along the same lines as in Proposition \ref{p14}:

\begin{proposition}
  \label{p15}For measurable sets $A,B \subset \mathbbm{R}^{n}$ such that $| A
  | < \infty ,  | B |  < \infty$, there are absolute constants $c_{1} ( n )
  >0,c_{2} ( n ) >0$ such that
  \[ \sup \{ J ( \varphi ,A,B ) :  \varphi \in L^{2} ( \mathbbm{R}^{n} ) , \|
     \varphi \| =1 \} \leqslant ( 1-c_{1} e^{-c_{2}   \eta_{n} ( A,B )} )^{2}
  \]
  where $\eta_{n} ( A,B )$ is given by the exponents in $\left( \ref{Jaming}
  \right)$ and/or $\left( \ref{Nazarov} \right)$ respectively.
\end{proposition}

Unfortunately, these constants are not yet optimal quantitatively, although
for some sets satisfying some geometrical properties (e.g. convexity) there
are good bounds (see {\cite{PJ2006}} for details). The lower bound
\[ \sup \{ J ( \varphi ,A,B ) :  \varphi \in L^{2} ( \mathbbm{R}^{n} ) , \|
   \varphi \| =1 \} \geqslant   ( 2 \pi )^{-2n}   | A |   | B |  e^{- (   \sup
   _{A} | x |^{2} ) - ( \sup_{B}   | k |^{2} )  } \]
shows that even for small domains the values of $J$ may be considerable, so
the exponential growth of the constants in
$\left(\ref{Jaming}\right),\left(\ref{Nazarov}\right)$ 
is no surprise at all.

\noindent \\
{\small\tt version: \$Id:: qmphspm.tex 2 2011-08-16 13:44:42Z kfp \$} \\
{\small\tt email: kp@scios.ch}

\end{document}